\documentclass[9pt,article,twocolumn]{IEEEtran} 

\usepackage{amsfonts,amsmath,amssymb}

\def\qed {{\unskip\nobreak\hfil\penalty25 \hskip2em\hbox{}\nobreak $\Box$ \parfillskip=0pt \finalhyphendemerits=0 \par\bigskip}}

\newcommand{\ef}{\mathbb{F}}
\newcommand{\tr}{\textup{Tr}}
\newcommand{\zZ}{\mathbb{Z}}
\newcommand{\cC}{\mathbb{C}}
\newcommand{\D}{\Delta}
\newcommand{\ga}{\gamma}

\newtheorem{thm}{Theorem}
\newtheorem{lem}[thm]{Lemma}

\begin{document}

\title{A characterization of two-weight projective cyclic codes}

\author{Tao Feng}


\date{\today}
\maketitle

\begin{abstract} We give necessary  conditions for a two-weight projective cyclic code to be the direct sum of two one-weight irreducible cyclic subcodes of the same dimension, following the work of Wolfmann and Vega. This confirms Vega's conjecture that all the two-weight cyclic codes of this type are the known ones in the projective case.
\end{abstract}

\begin{keywords}cyclic code, two-weight code, projective code, Gauss sum

\end{keywords}

\section{Introduction}
A projective code is a linear code such that the minimum weight of its dual code is at least $3$. A cyclic code is irreducible if its check polynomial is irreducible. Cyclic codes have important applications in digit communication, so it is of theoretical interest to determine the weight distribution of cyclic codes. We refer the reader to the classical textbooks \cite{ms} and \cite{vanLint} for basic facts on cyclic codes. When a cyclic code has exactly one nonzero weight, a nice characterization is given by Vega in \cite{vegaOne}, which we will describe later.

 The class of two-weight cyclic codes was extensively studied in \cite{c1,c5,c2,c3,wolfm,c6}. Schmidt and White obtained in \cite{sw} the necessary and sufficient conditions for an irreducible cyclic code to have at most two weights, and they also explored the connections with other combinatorial objects.

We fix the following notations throughout this note:
\begin{quote}
 {\it $p$ is a prime and $q$ is a power of $p$. Let $k\geq 1$ be a positive integer and $r=q^k$. Write $\D=\frac{r-1}{q-1}$. Let $\ga$ be a fixed primitive element of $\ef_r$. For an integer $a$, we use $h_a(x)$ for the minimal polynomial of $\ga^a$ over $\ef_q$.
 Let $\xi_{N}$ be the complex primitive $N$-th root of unity $e^{\frac{2\pi i}{N}}$ for any positive integer $N$.}
\end{quote}

 It was once conjectured that two-weight projective cyclic codes must be irreducible. Wolfmann \cite{wolfm} proved that this is true when $q=2$, but false when $q>2$. To be more specific, he proved that if $C$ is an [$n,k$] two-weight projective cyclic code over $\ef_q$ with $\gcd(n,q)=1$, then either

(1)  $C$ is irreducible, or

(2) $q>2$, $C$ is the direct sum of two one-weight irreducible cyclic subcodes of the same dimension, and $n=\lambda\frac{q^k-1}{q-1}$, where $\lambda|q-1$, $\lambda\ne 1$. Additionally, the two nonzero weights are $\lambda q^{k-1}$ and $(\lambda-1)q^{k-1}$.

After Wolfmann's work, several infinite families of two-weight cyclic codes of the second type were discovered, and a unified explanation is given by Vega in \cite{vega}. We state his result below.

\begin{thm}\cite[Theorem 10]{vega}\label{vgt} Let $a_1,a_2,v$ be integers such that  $a_1q^i\not\equiv a_2\pmod{q^k-1}$ for all $i\ge 0$, $v=\gcd(a_1-a_2,q-1)$, $a_2\in\zZ_{\D}^*$, and let $\tilde{a}_2$ be the inverse of $a_2$ in $\zZ_\D^*$. For an integer $\ell$ which divides $\gcd(a_1,a_2,q-1)$, we set $\lambda=\frac{(q-1)l}{\gcd(a_1,a_2,q-1)}$, $n=\lambda\D$ and $\mu=\frac{q-1}{\lambda}$. Suppose that at least one of the following two conditions holds:

(1) $p=2$, $k=2$, $v=1$, and $a_1$ is a unit in the ring $\zZ_\D$, or

(2) for some integer $j$, $1+\tilde{a}_2(a_1-a_2)\equiv p^j\pmod{v\D}$.

Then the following four assertions are true:

a) $h_{a_1}(x)$, $h_{a_2}(x)$ are the check polynomials for two different one-weight cyclic codes of length $n$ and dimension $k$;

b) $\mu|v$, and $\lambda>v/\mu$;

c) If $C$ is the cyclic codes with check polynomial $h_{a_1}(x)h_{a_2}(x)$, then $C$ is an $[n,2k]$ two-weight cyclic code with nonzero weights $\lambda q^{k-1}$ and $(\lambda-v/\mu)q^{k-1}$;

d) $C$ is a projective code if and only if $v=\mu$.
\end{thm}

\vspace{8pt}

It is not hard to show that under condition (1) in the above theorem, the code $C$ can not be projective. It is the purpose of this note to prove the following characterization of two-weight projective cyclic codes.

\begin{thm}\label{main}Let $C$ be an $[n,k]$  two-weight projective cyclic code over $\ef_q$ with $\gcd(n,q)=1$. Then $C$ is either irreducible, or the direct sum of two one-weight irreducible cyclic subcodes of the same dimension. Let $\ga$ be a fixed primitive element of $\ef_{q^k}$. In the latter case, $q>2$, and there exist integers $a_1,a_2$ such that

(1) $a_1\not\equiv a_2q^j\pmod{q^k-1}$ for any integer $j$;

(2) $a_1,a_2\in \zZ_{\D}^*$;

(3) $\gcd(a_1,a_2,q-1)=\gcd(a_1-a_2,q-1)=\frac{r-1}{n}$;

(4) $1+\tilde{a}_2(a_1-a_2)\equiv p^j\pmod{v\D}$  for some integer $j$, where $\tilde{a}_2$ is the inverse of $a_2$ in $\zZ_\D^*$.;

(5) the minimal polynomials of $\ga^{a_1},\ga^{a_2}$ over $\ef_q$, denoted by $h_{a_1}(x)$ and $h_{a_2}(x)$, have the same degree.

 Moreover, the product $h_{a_1}(x)h_{a_2}(x)$ is the check polynomial of the cyclic code $C$.
\end{thm}

\vspace{8pt}
We shall need the following characterization of one-weight cyclic code due to Vega \cite{vegaOne}. This will serve as the ingredient for the two-weight projective cyclic codes of the second type as described by Wolfmann.
\begin{lem}\label{vega1}
With the same notations as specified before, we have that $\gcd(a,\D)=1$ if and only if $h_a(x)$ is the check polynomial for a one-weight irreducible cyclic code.
\end{lem}

\section{Preliminaries}
In this section, we introduce the necessary backgrounds that we shall use in this note.

 \subsection{Group rings, characters}
 Let $G$ be a (multiplicatively written) finite abelian group with identity $1_G$. The group ring $\cC[G]$ is a ring, consisting of all the formal sums of the free basis $\{g|g\in G\}$ over $\cC$, with the multiplication extending that of $G$ by linearity and distributivity. We identify a subset $D$ of $G$ with the corresponding group ring element, namely $D=\sum_{g\in D}g\in\cC[G]$. For a group ring element $A=\sum_{g\in G}a_gg\in\cC[G]$, and an integer $m$, we use the notation $A^{(m)}:=\sum_{g\in G}a_gg^m$.
 Please refer to \cite{igr} for more details.

 We use $\widehat{G}$ for the character group of $G$, consisting of all the homomorphisms from $G$ to $\cC^\ast$. We use $\chi_0$ for the principal character of $G$, namely the homomorphism which maps each of $G$ to $1$. We have the following orthogonal relation:
 \[\sum_{\chi\in \widehat{G}}\chi(g)=\begin{cases}|G|,\;&\text{ if }g=1_G,\\0,\;&\text{ otherwise.}\end{cases}\]
The {\it inversion formula} follows from this relation:
{\it If $A=\sum_{g\in G}a_gg\in\cC[G]$, then }
\[
a_g=\frac{1}{|G|}\sum_{\chi\in\widehat{G}}\chi(A)\chi^{-1}(g),
\]
where $\chi(A)=\sum_{a\in G}a_g\chi(g)$. A direct consequence is that two elements $A,B\in\cC[G]$ are equal if and only if $\chi(A)=\chi(B)$ for each character $\chi$.\\

Let $\tr$ be the trace function from $\ef_{r}$ to $\ef _q$, and use $\tr_{\ef_q/\ef_p}$ for the trace function from $\ef_q$ to $\ef_p$. Let $\psi$ be the canonical character of $\ef_q$ defined by $\psi(x)=e^{\frac{2\pi i}{p}\tr_{\ef_q/\ef_p}(x)}$, $x\in\ef_q$. For each $a\in \ef_r$, we can define an additive character $ \psi_a$ of $\ef_r$ by $\psi_a(x)=\psi(\tr(ax))$ ($x\in\ef_r$). Let  $\varphi$ be the multiplicative character of $\ef_{r}^*$ which maps $\ga$ to $\xi_{r-1}$.

Each character of $\ef_r\times\ef_r$ is of the form $\psi_a\otimes\psi_b$ ($a,b\in\ef_r$) which is defined as\[\psi_a\otimes\psi_b((x,y))=\psi_a(x)\psi_b(y),\;\forall (x,y)\in\ef_r\times\ef_r. \]

Each character of the quotient group $\ef_r^*/\ef_q^*$ can be viewed as a multiplicative character of $\ef_r^*$ which is principal on $\ef_q^*$, and vice versa. Please refer to \cite{ff} for details.

\subsection{Gauss sums}
Let $\chi$ be a multiplicative character of $\ef_r^*$. The Gauss sum $G(\chi)$ is defined as \[G(\chi)=\sum_{x\in\ef_r^*}\psi(\tr(x))\chi(x).\]
Clearly, $G(\chi)=-1$ if $\chi=\chi_0$, and $|G(\chi)|=\sqrt{r}$ otherwise.  One important property of Gauss sums we shall use is the following
\[\psi(\tr(x))=\frac{1}{r-1}\sum_{\chi\in\widehat{\ef_r^*}}G(\chi)\chi^{-1}(x),\forall x\in\ef_r^*.\]
Please refer to \cite{ff} and \cite{gjs} for details on Gauss sums.

\subsection{Singer difference sets, multipliers}
Let $G$ be the quotient group $\ef_r^*/\ef_q^*$, and $D$ be the image of $\{x\in F_r^*|\tr(x)=1\}$ in $G$. It is well known that $D$ is the classical Singer difference set, which satisfies that
 \[
 DD^{(-1)}=q^{k-2}+q^{k-2}(q-1)G.
 \]
 For a nontrivial character $\chi$ of $G$, we have (\cite{yama})
 \[
 \chi(D)=-\frac{G(\chi)}{q}.
 \]
 In particular, $D$ is an invertible element in the ring $\cC[G]$. Under the character $\chi$, its inverse $D^{-1}$ has character value $\chi(D^{-1})=\frac{1}{\chi(D)}$, so $D^{-1}$ can be computed using the inversion formula. Similarly, we can show that $D^{(-1)}$ is invertible, since $\chi(D^{(-1)})=\overline{\chi(D)}$.

 If $t$ is an integer relatively prime with $|G|$, then $t$ is called a multiplier of $D$ if there exists $g\in G$ such that $D^{(t)}=Dg$. By \cite[Proposition 3.1.1]{Pott}, the only multipliers of the classical Singer difference sets are the powers of $p$.

 Please refer to \cite{Pott} or \cite{BJL} for more details on difference sets.

\subsection{Stickelberger's theorem}
Let $a$ be any integer not divisible by $r-1$. We use $L(a)$ for the least positive integer such that $L(a)\equiv a\pmod{r-1}$. Let $L(a)=a_0+a_1p+\cdots+a_{m-1}p^{m-1}$ ($r=p^m$) be the $p$-adic expansion of $L(a)$, with $0\leq a_i\leq p-1$ for each $i$. We define the function
\[
s(a)=\sum_{i=0}^{m-1}a_i.
\]
Let $\wp$ be the prime ideal in $\zZ[\xi_{r-1}]$ such that
\[
\xi_{r-1}\pmod{\wp}=\ga,
\]
and $\mathfrak{P}$ be the prime ideal in $\zZ[\xi_{r-1},\xi_p]$ lying over $\wp$. Stickelberger's theorem tells us that the highest power of $\mathfrak{P}$ that divides $G(\varphi^{-a})$ is $s(a)$, for any integer $a$ not divisible by $r-1$. Please refer to \cite{irc}.

\section{Proof of the main result}
\noindent{\bf Proof of Theorem \ref{main}:} Let $C$ be an $[n,k]$  two-weight projective cyclic code, and assume that $C$ is not  irreducible. Then by Wolfmann's result \cite{wolfm}, $C$ is  the direct sum of two one-weight irreducible cyclic subcodes of the same dimension; moreover, $q>2$, $n=\lambda\frac{r-1}{q-1}$, where $\lambda|q-1$ , $\lambda> 1$, and the two nonzero weights are $\lambda q^{k-1}$ and $(\lambda-1)q^{k-1}$. By Lemma \ref{vega1}, there exist $a_1,a_2\in \zZ_\D^*$ such that $h_{a_1}(x),h_{a_2}(x)$ have the same degree and their product $h_{a_1}(x)h_{a_2}(x)$ is the check polynomial of $C$. Since $x^n-1$ has no repeated root, $a_1\not\equiv a_2q^j\pmod{r-1}$ for any integer $j$. Therefore, $a_1,a_2$ satisfy the conditions (1), (2), (5), and we only need to prove that  $a_1,a_2$ satisfy the conditions (3) and (4).

Write $\beta:=\ga^{-1}$. In the trace form, the cyclic code $C$ consists of codewords of the form
\[c_{a,b}:=(\tr(a+b),\tr(a\beta^{a_1}+b\beta^{a_2}),\cdots,\tr(a\beta^{(n-1)a_1}+b\beta^{(n-1)a_2}))\]
with $a,b\in \ef_r$ (see \cite{vanLint,vega,wolfm}). We use $wt(c_{a,b})$ for the Hamming weight of the codeword $c_{a,b}$.

From $\ga^{a_1n}=\ga^{a_2n}=1$, we see that the length $n$ is divisible by $\frac{r-1}{\gcd(a_1,a_2,q-1)}$, i.e.,
 \[
 \ell:=\frac{n\cdot\gcd(a_1,a_2,q-1)}{r-1}=\frac{\lambda\cdot\gcd(a_1,a_2,q-1)}{q-1}
 \]
 is an integer.  Write $\tilde{a}_2$ for the inverse of $a_2$ in $\zZ_\D^*$. Also, we define
 \[
 v:=\gcd(a_1-a_2,q-1),\;w:=1+\tilde{a}_2(a_1-a_2).
 \]

\begin{lem}\label{lem_1} We have $\ell=1$,  $v=\gcd(a_1,a_2,q-1)=\frac{q-1}{n}$, and $\gcd(v,\D)=1$.
\end{lem}
\begin{proof}Since $C$ is projective, the existence of $0\leq i\leq j\leq n-1$, $\mu\in\ef_q^*$ such that $\tr(a\beta^{ia_1}+b\beta^{ia_2})=\mu \tr(a\beta^{ja_1}+b\beta^{ja_2})$ for any $a,b\in\ef_r$ would imply that $i=j$. To put it another way, if there exists $0\leq i\leq j\leq n-1$ and $\mu\in\ef_q^*$ such that $\beta^{ia_1}=\mu\beta^{ja_1}$, $\beta^{ia_2}=\mu\beta^{ja_2}$, then $i=j$. Since $\beta^{(i-j)a_1}=\mu$ is  in $\ef_q^*$, we have that $\D$ divides $a_1(i-j)$. Since $\gcd(a_1,\D)=1$, we have $i-j=u\D$ for some integer $u$. From $\beta^{(i-j)a_1}=\beta^{(i-j)a_2}$, we get $r-1|u\D (a_1-a_2)$, i.e., $q-1|u(a_1-a_2)$. Since $v=\gcd(a_1-a_2,q-1)$, $u$ is a multiple of $\frac{q-1}{v}$. If $\frac{q-1}{v}\D=\frac{r-1}{v}<n$, then the pair $(i,j)=(\frac{r-1}{v},0)$ would contradict the fact $C$ is projective. Hence $\frac{r-1}{v}\geq n$. On the other hand, $n$ is divisible by $\frac{r-1}{\gcd(a_1,a_2,q-1)}$, so
 \[
 \frac{r-1}{\gcd(a_1,a_2,q-1)}\leq n\leq \frac{r-1}{v}.
 \]
Since $v\geq \gcd(a_1,a_2,q-1)$, we immediately get $v=\gcd(a_1,a_2,q-1)$, $n=\frac{r-1}{v}$, $\ell=1$. Since $a_1$ is relatively prime with $\D$, we have $\gcd(v,\D)=1$.
\end{proof}

\vspace{8pt}

It follows from the above lemma that $a_1,a_2$ satisfy condition (3).
We break the remaining part of the proof into  a series of lemmas.\\

\begin{lem}\label{lem_2} The product  $G(\varphi^{\frac{q-1}{v}s})G(\varphi^{-\frac{q-1}{v}ws})$ is divisible by $r$, for any $s\not\equiv 0\pmod{v\D}$.
\end{lem}
\begin{proof}
We define the following subset $R\subset(\ef_r,+)\times(\ef_r,+)$ as in \cite{wolfm}:
\[
R:=\{(y\beta^{ia_1},y\beta^{ia_2})|0\leq i\leq n-1, y \in \ef_q^*\}.
\]
Let $a,b$ be two nonzero elements of $\ef_r^*$, and we associate the weight of $c_{a,b}$ to the value $\psi_a\otimes\psi_b(R)$ as follows.
 \begin{align*}
 \psi_a\otimes\psi_b(R)&=\sum_{i=0}^{n-1}\sum_{y\in \ef_q^*}\psi(y\tr(a\beta^{ia_1}+b\beta^{ia_2}))\\
 &=(-1)\cdot wt(c_{a,b})+(q-1)(n-wt(c_{a,b}))\\
 &=n(q-1)-qwt(c_{a,b}).
 \end{align*}
 This takes only two values, namely $-\lambda$ and $r-\lambda$, since $wt(c_{a,b})\in\{\lambda q^{k-1},(\lambda-1)q^{k-1}\}$. Now we compute it again in the McEliece way:
  \begin{align*}
 \psi_a&\otimes\psi_b(R)=\sum_{i=0}^{n-1}\sum_{y\in \ef_q^*}\psi(\tr(ay\beta^{ia_1})\psi(\tr(by\beta^{ia_2}))\\
 &=\frac{n}{r-1}\sum_{x\in\ef_r^*}\sum_{y\in \ef_q^*}\psi(\tr(ayx^{a_1})\psi(\tr(byx^{a_2}))\\
 &=\frac{1}{v(r-1)^2}\sum_{y\in F_q^*}\sum_{x\in F_r^*}\sum_{\chi,\phi\in\widehat{F_r^*}}G(\chi)G(\phi)\chi^{-1}(ayx^{a_1})\phi^{-1}(byx^{a_2})\\
 &=\frac{1}{v(r-1)^2}\sum_{\chi,\phi\in\widehat{F_r^*}}G(\chi)G(\phi)\chi^{-1}(a)\phi^{-1}(b)\\
 &\quad\quad\quad\quad\quad\quad\quad\quad\quad\cdot\sum_{x\in F_r^*}(\chi^{a_1}\phi^{a_2})(x)\sum_{y\in F_q^*}(\chi\phi)^{-1}(y)\\
 &=\frac{1}{v\Delta}\sum_{\chi,\phi}G(\chi)G(\phi)\chi^{-1}(a)\phi^{-1}(b),
 \end{align*}
 where the last summation is over those $\chi,\phi$ such that \[\chi\phi|_{\ef_q^*}=\chi_0,\quad \chi^{a_1}\phi^{a_2}=\chi_0.\] Here we use $\chi_0$ for the trivial character of $\ef_q^*$ and $\ef_r^*$, which we don't make a distinction. Write $\chi=\varphi^i$, $\phi=\varphi^j$ (recall that $\varphi$ is a fixed multiplicative character of order $r-1$). Then the above condition is translated to \[i+j\equiv 0\pmod{q-1},\quad a_1i+a_2j\equiv 0\pmod{r-1}.\]
 By exactly the same argument as in the proof of \cite[Lemma 7]{vega}, we see that this is equivalent to
 \[
 i=\frac{q-1}{v}t,\quad j=-iw,\quad0\leq t\leq v\D-1.
 \]
  We thus have
   \begin{align*}
 \psi_a&\otimes\psi_b(R)
 =\frac{1}{v\Delta}\sum_{t=0}^{v\D-1}G(\varphi^{\frac{q-1}{v}t})G(\varphi^{-\frac{q-1}{v}wt})
 \varphi^{-\frac{q-1}{v}t}(ab^{-w})\\
 &=\frac{1}{v\Delta}\sum_{t=1}^{v\D-1}G(\varphi^{\frac{q-1}{v}t})G(\varphi^{-\frac{q-1}{v}wt})
 \varphi^{-\frac{q-1}{v}t}(ab^{-w})+\frac{1}{v\D}\\
 &=\frac{1}{v\Delta}B_{ab^{-w}}+\frac{1}{v\D}\in\{-\lambda,\,r-\lambda\},
 \end{align*}
where we define
\[
B_x:=\sum_{t=1}^{v\D-1}G(\varphi^{\frac{q-1}{v}t})G(\varphi^{-\frac{q-1}{v}wt})
 \varphi^{-\frac{q-1}{v}t}(x),\quad x\in\ef_r^*.
 \]
We thus have $B_x\in\{-r,v\D r\}$ for any $x\in \ef_r^*$, by letting $b=1$, $a=x$ above. An easy consequence is that $r$ divides each of $B_x$, $x\in\ef_r^*$.

For each integer $1\leq s\leq v\D-1$, we compute that
 \begin{align*}
 &\frac{1}{r-1}\sum_{x\in\ef_r^*}
 B_x\varphi^{\frac{q-1}{v}s}(x)\\
 &=\frac{1}{r-1}\sum_{x\in\ef_r^*}\sum_{t=1}^{v\D-1}G(\varphi^{\frac{q-1}{v}t})G(\varphi^{-\frac{q-1}{v}wt})
 \varphi^{\frac{q-1}{v}(s-t)}(x)\\
 &=\frac{1}{r-1}\sum_{t=1}^{v\D-1}G(\varphi^{\frac{q-1}{v}t})G(\varphi^{-\frac{q-1}{v}wt})
 \sum_{x\in\ef_r^*}\varphi^{\frac{q-1}{v}(s-t)}(x)\\
 &=G(\varphi^{\frac{q-1}{v}s})G(\varphi^{-\frac{q-1}{v}ws}).
 \end{align*}
 Since $\gcd(r-1,r)=1$ and $r$ divides each of $B_x$, $x\in\ef_r^*$, this product is divisible by $r$. By the properties of Gauss sums, it also has modulus $r$.\\
\end{proof}

\noindent {\it Remark:} The subset $R$ here is a partial difference set in the group $(\ef_r,+)\times(\ef_r,+)$, c.f. \cite{ck}.\\

\begin{lem}\label{lem_3}  We have $G(\varphi^{\frac{q-1}{v}ws})=G(\varphi^{\frac{q-1}{v}s})\eta_s$ for some root of unity $\eta_s$, for any $s\not\equiv 0\pmod{v\D}$.
\end{lem}
\begin{proof}
Define the set $X$ as
\[
\left\{\pm\xi_{r-1}^iG(\varphi^{\frac{q-1}{v}s})G(\varphi^{-\frac{q-1}{v}ws})|i,s\in \zZ,s\not\equiv 0\pmod{v\D}\right\}.
\]
Each element in the set $X$ is divisible by $r$ as well as has modulus $r$ by  Lemma \ref{lem_2}.  We show that this set is fixed by $\text{Gal}(\mathbb{Q}(\xi_p,\xi_{r-1})/\mathbb{Q})$, so that each element of $X$ has all its algebraic conjugates in $X$. Let $\sigma_{u_1,u_2}$ be the automorphism which maps $\xi_p$ to $\xi_p^{u_1}$, $\xi_{r-1}$ to $\xi_{r-1}^{u_2}$, where $\gcd(p,u_1)=1$, $\gcd(r-1,u_2)=1$. For each multiplicative character $\chi$, we have
\[
\sigma_{u_1,u_2}(G(\chi))=\chi^{-u_2}(u_1)G(\chi^{u_1}).
\]
 The verification is now routine.

Now we have shown that each element of $X$ is divisible by $r$ and all its algebraic conjugates have modulus $r$. By a classical result of Kronecker (see \cite{irc}), it is equal to $r$ times a root of unity. Since
\[
G(\chi)\overline{G(\chi)}=r,\; G(\chi^{-1})=\chi(-1)\overline{G(\chi)}
\]
for any nonprincipal character $\chi$, we get the claim.\\
\end{proof}

In the case $s$ is a multiple of $v$, the above lemma shows that $G(\chi^w)=G(\chi)\eta_\chi$ for some root of unity $\eta_\chi$, where $\chi$ is any nonprincipal character of $\ef_r^*$ which is principal on $\ef_q^*$.\\

\begin{lem}\label{lem_4}   $w$ is a power of $p$ modulo $\D$, i.e., $a_1\equiv a_2p^j\pmod{\D}$ for some integer $j$.
\end{lem}
\begin{proof}
 Let $G$ be the quotient group $\ef_r^*/\ef_q^*$, and $D$ be the classical Singer difference set in $G$. Recall that for a nontrivial character $\chi$ of $G$, we have (\cite{yama})
 \[\chi(D)=-\frac{G(\chi)}{q}.\]
 By the previous lemma, we see that $\chi(D^{(w)})=\sum_{d\in D}\chi(d^w)=\chi^w(D)$ and $\chi(D)$ differ by a root of unity. We shall show that $w$ is a multiplier of $D$, and hence by \cite[Proposition 3.1.1]{Pott}, $w$ is a power of $p$ modulo $\D$. The following argument is a standard trick in the proofs of various multiplier theorems.

 First define $F:=D^{(w)}D^{(-1)}-q^{k-2}(q-1)G\in \zZ[G]$. For each nonprincipal character $\chi$ of $G$, we have \[\chi(F)=\frac{G(\chi^w)G(\chi^{-1})}{q^2}=q^{k-2}\eta_\chi,\]
 for some root of unity $\eta_\chi$. If $\chi_0$ is the principal character, then $\chi_0(F)=|D|^2-q^{k-2}|G|=q^{k-2}\eta_{\chi_0}$, with $\eta_{\chi_0}=1$.

 Let $g$ be an element of $G$ which has nonzero coefficient in $F$. By the inversion formula, this coefficient is
 \[
 \frac{1}{|G|}\sum_{\chi\in\widehat{G}}\chi(F)\chi^{-1}(g)=q^{k-2}\frac{\sum_{\chi\in\widehat{G}}\eta_\chi\chi^{-1}(g)}{|G|}\in\zZ.
 \]

 Now since $q^{k-2}$ and $|G|$ are relatively prime, $\frac{\sum_{\chi}\eta_\chi\chi^{-1}(g)}{|G|}$ is a rational integer. This number has modulus not exceeding $1$ and not equal to zero, so is either $1$ or $-1$. We deduce that $\eta_\chi=\epsilon\chi(g)$ for each $\chi\in\widehat{G}$, where $\epsilon\in\{\pm 1\}$. Since $\eta_{\chi_0}=1$, we must have $\epsilon=1$. Now, for any $h\in G\setminus\{g\}$, its coefficient in $F$ is equal to
 \[
 \frac{1}{|G|}\sum_{\chi\in\widehat{G}}\chi(F)\chi^{-1}(h)=\frac{q^{k-2}}{|G|}\sum_{\chi\in\widehat{G}}\chi(gh^{-1})=0
 \]
 by the orthogonal relation.  We thus have $F=q^{k-2} g$. Comparing this with
 \[
 DD^{(-1)}=q^{k-2}+q^{k-2}(q-1)G,
 \]
 and using the fact that $D^{(-1)}$ is invertible in the ring $\cC[G]$,  we conclude that $D^{(w)}= Dg$. Hence $w$ is a multiplier of the difference set $D$. The claim then follows.\\
\end{proof}

Using Stickelberger's theorem, we  deduce from Lemma \ref{lem_3} that
\begin{equation}\label{dig_sum}
 s\left(\frac{q-1}{v}wt\right)=s\left(\frac{q-1}{v}t\right),\forall\, t\not\equiv 0\pmod{v\D}.
\end{equation}
The set of integers $x$'s relatively prime with $v\Delta$ that satisfy the above equation with $w$ replaced by $x$ is clearly closed under multiplication, so forms a multiplicative subgroup of $\zZ_{v\Delta}^*$ after modulo $v\Delta$. Moreover, $p$ is such an integer.

By Lemma \ref{lem_4}, we have $w\equiv p^j\pmod{\Delta}$ for some integer $0\leq j\leq m-1$ (here $r=p^m$). Since $v=\gcd(a_1-a_2,q-1)$, we naturally have $w\equiv 1\pmod{v}$, recalling that $w=1+\tilde{a}_2(a_1-a_2)$. In order to show that $w\equiv p^j\pmod{v\Delta}$, we only need to prove that $p^j\equiv 1\pmod{v}$, using the Chinese Remainder Theorem and the fact that $\gcd(v,\Delta)=1$.   \\

\begin{lem} $w$ is a power of $p$ modulo $v\D$.
\end{lem}
\begin{proof}Let $d=\gcd(j,m_0)$, where $q=p^{m_0}$, $r=p^m=p^{m_0k}$ and $j$ is the same as in the previous lemma. There exist positive integers $u_0,v_0$ such that $u_0j+v_0m_0\equiv d\pmod{m}$. Define $w_0:=w^{u_0}q^{v_0}\pmod{v\Delta}$. Then Eqn. \eqref{dig_sum} also holds with $w$ replaced by $w_0$, and
\[
w_0\equiv 1\pmod{v},\; w_0\equiv p^d\pmod{\Delta}.
 \]
We aim to show that $v$ divides $p^d-1$ and thus $p^j-1$. Write $m_0=r_0d$. Since $v|q-1$, we only need to deal with the case $r_0>1$. The case $d=1$ and $d>1$ can be dealt with in much the same way, but for clarity we deal with them separately.\\

First we consider the case $d=1$. Write
\[
\frac{q-1}{v}=a_0+a_1p+\cdots+a_{r_0-1}p^{r_0-1},\;0\leq a_i\leq p-1.
\]
If $a_0=a_1=\cdots=a_{r_0-1}$, then $\frac{q-1}{p-1}$ divides $\frac{q-1}{v}$, so $v|p-1$ and we are done. Assume that the $a_i$'s are not all equal. Let $i_0$ be such that $a_{i_0}$ is smallest among all $a_i$'s and $a_{i_0}<a_{i_0+1}$, where the subscript is taken modulo $r_0$; similarly, let $j_0$ be such that  $a_{j_0}$ is largest among all $a_i$'s and $a_{j_0}>a_{j_0+1}$. This guarantees that $a_{j_0+1}+1<p$, $a_{i_0+1}-1\geq 0$. Set $x:=j_0-i_0\pmod{m}$. We have
\[
\frac{q-1}{v}w_0^{x}\Delta \pmod{r-1}=\frac{q-1}{v}\Delta,
\]
which is equal to $\sum_{i=0}^{k-1}(a_0+a_1p+\cdots+a_{r_0-1}p^{r_0-1})p^{r_0i}$.

Take $t=\Delta+vp^{i_0}\delta $, with $0\leq \delta\leq p-1$ to be specified later. Then $\frac{q-1}{v}t$ is congruent to the sum of
\begin{align*}
&(a_0+a_1p+\cdots+p^{i_0}(a_{i_0}-\delta)+\cdots+a_{r_0-1}p^{r_0-1}),\\
&(a_0+a_1p+\cdots+p^{i_0}(a_{i_0}+\delta)+\cdots+a_{r_0-1}p^{r_0-1})p^{r_0},
\end{align*}
and $\sum_{i=2}^{k-1}(a_0+a_1p+\cdots+a_{r_0-1}p^{r_0-1})p^{r_0i}$ modulo $r-1$. Similarly, $\frac{q-1}{v}w_0^{x}t $ is congruent to the sum of
\begin{align*}
&(a_0+a_1p+\cdots+p^{j_0}(a_{j_0}-\delta)+\cdots+a_{r_0-1}p^{r_0-1}),\\
&(a_0+a_1p+\cdots+p^{j_0}(a_{j_0}+\delta)+\cdots+a_{r_0-1}p^{r_0-1})p^{r_0},
\end{align*}
and $\sum_{i=2}^{k-1}(a_0+a_1p+\cdots+a_{r_0-1}p^{r_0-1})p^{r_0i}$ modulo $r-1$. When $\delta=0$, the digit sums $s(\frac{q-1}{v}w_0^{j_0-i_0}t)$ and $s(\frac{q-1}{v}t)$ are both equal to $s(\frac{q-1}{v}\Delta)$, which we denote by $N$.

Take $\delta$ to be the smallest number among $a_{i_0}+1$, $a_{j_0}+1$, $p-a_{i_0}$ and $p-a_{j_0}$, which must be one of $a_{i_0}+1$ and $p-a_{j_0}$. If $\delta$ occurs only once among these four numbers, then one of   $s(\frac{q-1}{v}w_0^{x}t)$ and $s(\frac{q-1}{v}t)$ is equal to $N$, and the other differs from $N$ by  $p-1$. To be specific, by writing out their $p$-adic expansions explicitly, we see that:  if $\delta=a_{i_0}+1$, then $s(\frac{q-1}{v}w_0^{x}t)=N$, $s(\frac{q-1}{v}t)=N+p-1$; if $\delta=p-a_{j_0}$, then $s(\frac{q-1}{v}w_0^{x}t)=N-(p-1)$, $s(\frac{q-1}{v}t)=N$.  This is a contradiction to the fact that they are equal. Hence it must be that $\delta=a_{i_0}+1=p-a_{j_0}$. In this case, $a_{i_0}+\delta=a_{i_0}+p-a_{j_0}<p$, $a_{j_0}-\delta=a_{j_0}-a_{i_0}-1\geq 0$. Hence, for this choice of $\delta$, $s(\frac{q-1}{v}t)=N+p-1$, and $s(\frac{q-1}{v}w_0^{x}t)=N-(p-1)$ by writing out their $p$-adic expansions explicitly, which are not equal. This is also a contradiction.

We conclude that the $a_i$'s are all equal, so $\frac{q-1}{p-1}$ divides $\frac{q-1}{v}$ and $v$ divides $p-1$. This proves the case $d=1$.\\

Now let's assume that $d>1$. Write
\[
\frac{q-1}{v}=a_0+a_1p^d+\cdots+a_{r_0-1}p^{(r_0-1)d},\;0\leq a_i\leq p^d-1,
\]
and for each $i$ we write
\[
a_i=a_{i,0}+a_{i,1}p+\cdots+a_{i,d-1}p^{d-1},\; 0\leq a_{i,j}\leq p-1.
\]
The idea is to show that $a_{0,i}=a_{1,i}=\cdots=a_{r_0-1,i}$ for each $0\leq i\leq d-1$, using the same argument as in the case $d=1$. Because the proof is a slight modification of the $d=1$ case, we don't include all the details here.

Take any $k_0\in\{0,\cdots,d-1\}$ such that  $a_{i,k_0}$, $0\leq i\leq r_0-1$, are not all the same. Let $i_0$ (resp. $j_0$) be such that $a_{i_0,k_0}$ (resp. $a_{j_0,k_0}$) is smallest (resp. largest) among all $a_{i,k_0}$'s and $a_{i_0,k_0}<a_{i_0+1,k_0}$ (resp. $a_{j_0,k_0}>a_{j_0+1,k_0}$), where the first subscript is taken modulo $r_0$. Let $N$ be the digit sum  $s(\frac{q-1}{v}\Delta)$ as before. Let $x:=j_0-i_0\pmod{m}$.

Take $t=\Delta+vp^{i_0d+k_0}\delta $, with $\delta$ the smallest among the numbers $a_{i_0,k_0}+1$, $a_{j_0,k_0}+1$, $p-a_{i_0,k_0}$ and $p-a_{j_0,k_0}$.  Assume that $\delta$ occurs only once among these four values. If $\delta=a_{i_0,k_0}+1$, then $s\left(\frac{q-1}{v}w_0^{x}t\right)=N$, and
$s\left(\frac{q-1}{v}t\right)=N+(l_0+1)(p-1)$, where $l_0$ counts the number of consecutive $0$'s immediately following $a_{i_0,k_0}p^{di_0+k_0}$ in the $p$-adic expansion of $\frac{q-1}{v}\Delta+q\cdot p^{i_0d+k_0}\delta$. Since $a_{i_0+1,k_0}>a_{i_0,k_0}\geq 0$, we have $0\leq l_0\leq d-1$. If $\delta=p-a_{j_0,k_0}$, then $s\left(\frac{q-1}{v}w_0^{x}t\right)=N-(l_1+1)(p-1)$, and
$s\left(\frac{q-1}{v} t\right)=N$, where $l_1$ counts the number of consecutive $p-1$'s immediately following $a_{j_0,k_0}p^{dj_0+k_0}$ in the $p$-adic expansion of $\frac{q-1}{v}\Delta-p^{i_0d+k_0}\delta$. Similarly, $0\leq l_1\leq d-1$. In both cases, the two digit sums are not equal, which is a contradiction. Hence $\delta=a_{i_0,k_0}+1=p-a_{j_0,k_0}$, but in this case,
$s\left(\frac{q-1}{v}w_0^{x}t\right)<N$, and $s\left(\frac{q-1}{v} t\right)>N$ by arguing as in the $d=1$ case. This is also a contradiction.

We conclude that $a_{0,i}=a_{1,i}=\cdots=a_{r_0-1,i}$ for each $0\leq i\leq d-1$, and so $a_0=a_1\cdots=a_{r_0-1}$. It follows that $\frac{q-1}{p^d-1}$ divides $\frac{q-1}{v}$, i.e., $v$ divides $p^d-1$. Hence $p^j\equiv 1\pmod{v}$, and the desired congruence $w\equiv p^j\pmod{v\D}$ follows.\\
\end{proof}

 This proves (4), and completes the proof of Theorem \ref{main}. \qed

\section{Conclusion}
In this note, we give necessary  conditions for a two-weight projective cyclic code to be the direct sum of two one-weight irreducible cyclic subcodes of the same dimension, following the work of Wolfmann and Vega.
In particular, these necessary conditions are also sufficient in view of Theorem \ref{vgt}. Based on computer search evidence,  Vega \cite{vegac} conjectured that all two-weight cyclic codes which are the direct sum of two one-weight irreducible cyclic subcodes of the same dimension are the known ones. Our result confirms the conjecture in the projective case.

Finally, we mention that the weight distributions of cyclic codes with two zeros $\ga^{d}$, $\ga^{d+D}$ recently are studied in several papers (\cite{d1,c2,d2,d3}), where $d|q-1,\,D=\frac{r-1}{e},\,ed|q-1$ and $ed>1$. These codes will never be projective. It is of interest to determine the two-weight codes  among all such cyclic codes. This will provide further verification to Vega's conjecture.

\section*{Acknowledgments}

The author is with the Department of Mathematics, Zhejiang University, China. This research is partially supported by Fundamental Research Funds for the Central Universities of China, Zhejiang Provincial Natural Science Foundation (LQ12A01019), National Natural Science Foundation of China (11201418).


\end{document}